\documentclass[12pt,a4paper]{article}

\usepackage{latexsym}
\usepackage{amssymb}
\usepackage{amsmath}

\usepackage{color}
\usepackage{harvard}
\renewcommand{\cite}[1]{\citeasnoun{#1}}
\newtheorem{theorem}{Theorem}

\newtheorem{remark}[theorem]{Remark}

\newtheorem{assumption}[theorem]{Assumption}

\newcommand{\rref}[1]{{\rm (\ref{#1})}}

\newcommand{\cF}{{\cal F}}
\newcommand{\cL}{{\cal L}}

\newcommand{\cO}{{\cal O}}

\newcommand{\cX}{{\cal X}}

\newcommand{\EE}{\mathbb{E}}

\newcommand{\E}{\EE}

\newcommand{\condexp}[1]{ \E \left[ \left.#1 \right| {\cal F}_{t} \right] }

\newcommand{\qed}{\mbox{ }~\hfill~$\Box$ \vspace{1ex} }

\newenvironment{proof}{\noindent{\sc Proof : }}{ \qed }

\newenvironment{rmenumerate}
  {\begin{enumerate}}
  {\end{enumerate}}

\newcommand{\frederik}[1]{{  #1}}
\begin{document}

\title{Existence of Financial Equilibria in Continuous Time with Potentially Complete Markets}

\author{Frank Riedel\thanks{We gratefully acknowledge financial support by the German Research Foundation (DFG) via the International Research Training Groups ``Economic Behavior and Interaction Models''  and ``Stochastics and Real World Models''. } \and Frederik Herzberg}
\date{Institute of Mathematical Economics\\
Bielefeld University}
\maketitle

\begin{abstract}
\noindent We prove that in smooth Markovian continuous--time economies with potentially complete asset markets, Radner equilibria with endogenously complete markets exist.
\end{abstract}

\medskip
{\it \hspace*{0.6cm} JEL subject classification:} D52, D53, G12

\medskip
{\it \hspace*{0.6cm} 2010 Mathematics Subject Classification:} 91B50, 91G80

\medskip
{\it \hspace*{0.6cm} Keywords:} Potentially complete market, Continuous-time financial market, Radner equilibrium, It\^o diffusion

\section*{Introduction}

The hallmark of economics is still the general theory of competitive markets as expressed masterfully in the work of Arrow and Debreu. While this theory can be considered as complete, its extension to competitive markets under uncertainty in continuous time remains still imperfect. In discrete models, it is well known that for potentially complete markets of real assets, one generically has a Radner equilibrium with endogenously generated complete markets that implement the efficient allocation of the corresponding  Arrow--Debreu equilibrium, see \cite{MagillShafer85} or \cite{MagillQuinzii98}, Theorem 25.7.

  \cite{AndersonRaimondo08} prove  a version of this theorem for  specific continuous--time economies where endowments and dividends are smooth functions of Brownian motion and time, and  agents have time--separable expected utility functions.
They establish their result with the help of nonstandard analysis, an intriguing approach to analysis and stochastics via mathematical logic that allows, e.g., to work with infinitely large and infinitesimally small    numbers, and to identify Brownian motion with a random walk of infinite length and infinitesimally small time steps. We believe that such an important theorem deserves a standard proof -- we provide it here\footnote{In independent work,
\frederik{\cite{Hugonnieretal2012}} prove a remarkably similar result to ours. In their first version \frederik{(which eventually developed into \cite{Hugonnieretal2011})}, these authors \emph{assumed} that equilibrium state prices are analytic. Their  last version, which appeared after our paper was available, proves this assumption on endogenous objects, as we do and did in all versions of our paper. We think that our treatment is clearer from an economic point of view, we are more to the point, and we have exact references. We thus hope that our paper provides an interesting reading for our readers. A very elegant generalization of the crucial mathematical part of the analysis can be found in a recent working paper by \cite{KramkovPredoiu11}. This paper is motivated by lectures on General Equilibrium Theory that one of the authors (Frank Riedel) gave at Carnegie Mellon University in 2008.
}.

At the same time, we extend the result to more general classes of state variables.
Many finance models nowadays rely on more general diffusions; prominent examples include the stochastic volatility models, where the volatility of the risky asset is a mean--reverting process as in \cite{Heston93},  term structure models like \cite{Vasicek77} or more generally affine term structure models as in \cite{DuffiePanSingleton00}. It is thus important to have sound equilibrium foundations for such models as well.

The paper is set up as follows. The next section describes a smooth continuous--time Markov economy where all relevant functions are analytic on the open interior of their domain. In this paper, the term ``analytic'' (=real analytic) refers to infinitely differentiable functions that can be written locally as an infinite power series\footnote{Our reference is \cite{KrantzParks02}.}.
Then, we formulate our main theorem on existence of a Radner equilibrium with endogenously dynamically complete markets. The proof is split in several steps. We first recall Dana's \citeyear{Dana93} result on existence of an Arrow--Debreu equilibrium and show that in our setup, allocation and prices are analytic functions of time and the state variable. The natural candidates for security prices are the expected present values of future dividends. We show that these can also be expressed as analytic functions  of time and the state variable if natural assumptions on the coefficients of the diffusion are satisfied. On the one hand, if one has a closed--form version of the state variable's transition density, the result holds true. This is straightforward to check in the case of Brownian motion, or mean--reverting diffusions, e.g. From an abstract point of view, it is better to have conditions on the primitive of the model that ensure such a nice transition density. We state sufficient conditions on the drift and dispersion coefficients of our state variable for such a result.

The analyticity of security prices allows us to extend the local independence assumption on terminal dividends to security prices, proving dynamic completeness, as in \cite{AndersonRaimondo08}.  The implementation of the Arrow--Debreu equilibrium as a Radner equilibrium is then standard.

\section{A Diffusion Exchange Economy with Potentially Complete Asset Markets}

In this section, we set up an exchange economy in continuous time where the relevant information is generated by a diffusion $X=(X_t)_{t\in[0,T]}$ with values in $\mathbb R^K$. It is well known that one needs at least $K+1$ financial assets to span a dynamically complete market. We thus assume that this necessary condition is satisfied. The market is thus \emph{potentially complete}. Below, we  show that in sufficiently smooth economies a Radner equilibrium with dynamically complete markets exists.

\subsection{The State Variables}
 Let $W$ be a $K$--dimensional Brownian motion on a complete probability space $(\Omega, \cF, P)$. Denote by $\left(\cF_t\right)_{t \ge 0}$ the filtration generated by $W$ augmented by the null sets. We assume that the relevant economic information can be described by the state of a diffusion process $X$ with values in $\mathbb R^K$ given by
\begin{equation}
  \label{EqnX}
X_0=x, dX_t = b(X_t) dt + \sigma(X_t) dW_t \,,
\end{equation}
for an initial state $x \in \mathbb R^K$ and measurable functions
$$b : \mathbb R^K \to \mathbb R^K$$ and
$$\sigma : \mathbb R^K \to \mathbb R^{K \times K}$$ that are called the drift and dispersion function, resp. We let
$$a(x):=\sigma(x) \sigma(x)^T$$ be the diffusion matrix.

\begin{assumption}
  \label{AssSDE}
  \begin{enumerate}
    \item $b$ and $\sigma$ are Lipschitz--continuous:
there exist $L,M>0$ such that for all $x,y \in \mathbb R^K$
$$\left\| b(x)-b(y) \right\| \le L \left\|x-y\right\|, \left\| \sigma(x)-\sigma(y) \right\| \le L \left\|x-y\right\|$$
\item
The diffusion matrix satisfies the uniform ellipticity condition
\begin{equation}\label{EqnUniformElliptic}
   \left\| x \cdot a(x) x \right\| \ge \epsilon \left\|x\right\|^2
 \end{equation} for some $\epsilon>0$.
  \end{enumerate}
\end{assumption}

 Part 1 of the assumption ensures that the stochastic differential equation has a unique strong solution and so our state variable is well--defined. The uniform ellipticity condition \rref{EqnUniformElliptic} ensures that there is enough volatility in every state and the diffusion does not degenerate to a locally deterministic process; in particular, it ensures that the distribution of $X$ has full support, see \cite{StroockVaradhan72}.

\begin{assumption}\label{AssSDECoefficientsBoundedHoelderContinuous}
$b$ and $\sigma$ as well as its derivatives are bounded, H\"{o}lder--continuous, and analytic functions.
\end{assumption}

\subsection{Commodities and Agents}

There is one physical commodity in the economy.
Our agents consume a flow $(c_t)_{0 \le t < T}$ and a lump-sum $c_T$ of that commodity at terminal time $T$. We introduce the measure $\nu=dt \otimes \delta_T$, the product of  the Lebesgue measure on $[0,T]$ and the Dirac measure on $\{T\}$. This allows us to model the consumption plans succinctly as one process $c=(c_t)_{0\le t \le T}$ in the following way.
 The commodity  space $\cX$ consists of $p$--integrable consumption rate processes and a $p$--integrable terminal  lump sum consumption for some $p\ge 1$,
$$\cX= L^p\left( \Omega \times [0,T], \cO, P \otimes \nu\right)\,.$$
The consumption set is the positive cone $\cX_+$.
We will use occasionally the dual space of $\cX$ that we shall call the price space
$$\Psi = L^q\left( \Omega \times [0,T], \cO, P \otimes \nu\right)$$
for $q$ with $1/q+1/p=1$.

There are $i=1,\ldots,I$ agents
with time--separable expected utility preferences of the form
$$U^i(c)= \E \int_0^T u^i\left( t, c_t\right) \nu(dt)$$
for a period utility function
$$u^i : [0,T] \times \mathbb R_+ \to \mathbb R \,.$$

\begin{assumption}
\label{AssUtilityGeneral}
  The period utility functions $u^i$ are continuous on $[0,T] \times \mathbb R_+$ and analytic on $(0,T)\times \mathbb R_{++}$. They are  differentiably strictly increasing and differentiably strictly concave in consumption on $[0,T]\times\mathbb R_{++}$, i.e.
  $$ \frac{\partial u^i}{\partial c} \left(t,c\right) >0, \frac{\partial^2 u^i}{\partial c^2} (t,c) < 0 \,.$$ They satisfy the Inada conditions
  $$\lim_{c\downarrow 0}  \frac{\partial u^i}{\partial c} \left(t,c\right)=\infty$$
  and
  $$\lim_{c\to\infty} \frac{\partial u^i}{\partial c} \left(t,c\right)=0$$
  uniformly in $t\in[0,T]$.
\end{assumption}

\begin{assumption}\label{AssEntitlementAnalytic}
Each agent comes with a $P\otimes \nu$--strictly positive entitlement\footnote{We use the word ``entitlement'' here to distinguish it from the total initial endowment used below which is the sum of the entitlement and the dividends of assets initially owned by the agent.} $e^i \in \cX_+$ that can be written as a function of the state variables:
 $$e^i_t= e^i\left( t, X_t\right) $$ for
 continuous functions $ e^i : [0,T] \times \mathbb R^{K} \to \mathbb R$, $i=1,\ldots,I$.
  The functions $e^i$ are analytic on $(0,T)\times \mathbb R^K.$
\end{assumption}

\subsection{The Financial Market}

\begin{assumption}\label{AssDividendAnalytic}
There are $K+1$ financial assets. These are \emph{real} assets in the sense that they pay dividends in terms of the underlying physical commodity. The assets' dividends can be written as
$$ A^k_t= g^k\left(t,X_t\right), t \in [0,T]$$
for continuous functions $ g^k : [0,T] \times \mathbb R^{K} \to \mathbb R_+$, $k=0,\ldots,K$. As for consumption processes, we interpret dividends as a flow on $[0,T)$ plus a lump sum payment at time $T$.

 The dividends belong to the consumption set, $A^k \in\cX_+$.  The functions $g^k$ are analytic on $(0,T)\times \mathbb R^K.$
  Asset $0$ is a real zero--coupon bond with maturity $T$; it has no intermediate dividends, i.e. $A^0_t=0$ for $t<T$\footnote{We can also work with intermediate dividends. In that case, an additional small detour  is necessary in  order to construct a suitable num\'{e}raire asset. As this part is not at the heart of the present analysis, we do not present this generalization here. The argument is available from the authors. }.
\end{assumption}

Agent $i$ owns initially $n^i_{k}\ge 0$ shares of asset $k$. Without any trade, the agent is thus endowed with his individual endowment
$$\varepsilon^i_t=e^i_t + n^i \cdot A_t\,.$$ We denote by $N_k=\sum_{i=1}^I n^i_k$ the total number of shares in asset $k$. The aggregate endowment of agents is then
$$\varepsilon_t = \sum_{i=1}^I e^i_t + \sum_{k=0}^K N_k A^k_t = \sum_{i=1}^I \varepsilon^i_t  \,.$$

A consumption price process is a positive It\^{o} process $\psi$.
A (cum--dividend) security price for asset $k$ is a nonnegative It\^{o} process $S^k=\left(S^k_t\right)_{0\le t\le T}$. We interpret $S^k$ as the nominal price of the asset $k$.
We denote by $$G^k_t = S^k_t + \int_{[0,t)} A^k_s  \psi_s \nu(ds) ,\qquad(0\le t \le T)
$$
 the (nominal) gain process for asset $k$.
Note that by no arbitrage we must have $S^k_T=A^k_T$ at maturity.

A portfolio process is a predictable process $\theta$ with values in $\mathbb R^{K+1}$ that is $G$--integrable, i.e.
the stochastic integrals $\int_0^t \theta^k_u dG^k_u$ are well--defined. The value of such a portfolio is
$V_t=\theta \cdot S$.

We call a portfolio admissible (without reference to an agent) if its value process is bounded below by a martingale.
This admissibility condition rules out doubling strategies\footnote{Anderson and Raimondo use a martingale condition to rule out such strategies. This requires to impose a martingale condition on potential security prices. As this martingale property is a consequence of equilibrium, we prefer not to impose this assumption ex ante. Nevertheless, either way works here.}.

A portfolio is \emph{admissible for agent $i$} if its present value plus the present value of the agent's endowment is nonnegative, or
$$V_t +  \condexp{\int_{t+}^T e^i_s \psi_s \nu(ds) } \ge 0 \,.$$
 Note that this implies $V_T\ge 0$ for the terminal value of the portfolio.

A portfolio $\theta$ finances a consumption plan $c\in\cX_+$ for agent $i$ if $\theta$ is admissible for agent $i$ and the intertemporal budget constraint is satisfied for the associated value process $V$:
$$V_t = n^i \cdot S_0 + \int_0^t \theta_u dG_u + \int_0^t \left(e^i_u-c_u\right) \psi_u \nu(du) \,.$$ We then call the portfolio/consumption pair $(\theta, c)$ $i$--feasible.
More generally, we say that a portfolio $\theta$ finances a net consumption plan $z \in \cX$ if its value process satisfies
$$V_t = V_0 + \int_0^t \theta_u dG_u + \int_0^t \left(e^i_u-c_u\right) \psi_u \nu(du) \,.$$

A Radner equilibrium consists of asset prices $S$, a consumption price $\psi$, portfolios $\theta^i$ and consumption plans $c^i\in\cX_+$ for each agent $i$ such that
$\theta^i$ is admissible for agent $i$ and finances $c^i$, $c^i$ maximizes agent $i$'s utility over all such $i$--feasible portfolio/consumption pairs, and markets clear, i.e.
$\sum_{i=1}^I c^i = \varepsilon $ and $\sum_{i=1}^I \theta^i = N$.

Our way to a Radner equilibrium with dynamically complete markets will lead over the intermediate step of an Arrow--Debreu equilibrium. For the existence of such an equilibrium, the following assumption is, in general, necessary\footnote{Assumption \ref{margfel} cannot be weakened in general. Assume
that there is only one agent. Then, to establish a no-trade
equilibrium in the Arrow-Debreu sense, it is necessary to find a
price
$\psi \in L$ that separates the endowment $e$ from the set
$G =\{ d \in L ; U(d) \geq U(e) \}$ of consumption streams preferred to $e$.
The only candidate in a smooth model like this one for such a
price process is the marginal felicity $\frac{\partial}{\partial
c} u(t,e_t)$. If it is not square-integrable, then there exists no
equilibrium. For more on the necessity of Assumption
\ref{margfel}, the reader may consult the overview of
\cite{MasColellZame91}, especially Example 6.5, and the paper of  \cite{AraujoMonteiro91}, where it is shown  that an
equilibrium does generically not exist if one does not have a
condition on the integrability of marginal felicities like
Assumption \ref{margfel}.}:

\begin{assumption}\label{margfel}
\begin{rmenumerate}
  \item For each agent, the marginal utility of his endowment belongs to the price space $\Psi$:
$$  \frac{\partial}{\partial c} u^i(t,\varepsilon^i_{t}) \in \Psi \,.$$
\item
Aggregate endowment $\varepsilon$ is bounded and bounded away from zero.
\end{rmenumerate}
\end{assumption}

Note that Assumption \ref{margfel} ii) implies Part i).

If the assets are linearly dependent, there is no hope to span a dynamically complete market. To exclude this, we follow \cite{AndersonRaimondo08} and impose a full rank condition on terminal payoffs:

\begin{assumption}\label{AssIndependentPayoffs}
  On a nonempty open set $V \subset \mathbb R^K$, the dividend of the zero--th asset is strictly positive at maturity,
  $$g^0(T,x)>0, \qquad (x\in V)\,.$$
  The functions $h^k: x \mapsto \frac{g^k(T,x)}{g^0(T,x)}$ are continuously differentiable on $V$ for $k=1,\ldots, K$ and the Jacobian matrix
  $$Dh (x) = \begin{pmatrix}
               \frac{\partial h^1(T,x)}{\partial x_1} & \ldots & \frac{\partial h^1(T,x)}{\partial x_K} \\
               \vdots & \ddots & \vdots \\
               \frac{\partial h^K(T,x)}{\partial x_1} & \ldots & \frac{\partial h^K(T,x)}{\partial x_K} \\
             \end{pmatrix}
  $$ has full rank on $V$.
\end{assumption}

\section{Existence of Radner Equilibrium with Dynamically Complete Markets}

We are now in the position to state our main result. We call the  market given by the asset prices $S$, dividends $A$, and consumption price $\psi$ dynamically complete if  every net consumption plan $z \in \cX$ can be financed by an admissible portfolio $\theta$ in the sense that its value process satisfies
$$V_t=V_0+ \int_0^t \theta_u dG_u + \int_0^t z_u \psi_u \nu(du) \,.$$

\begin{theorem}\label{ThmMain} Under Assumptions 1 to 5, \ref{margfel} ii), and 7,  there exists a Radner equilibrium $\left(S,\psi,\left(\theta^i,c^i\right)_{i=1,\ldots,I}\right)$ with
   a dynamically complete market $(S,A,\psi)$; the prices and dividends are linked by the present value relation
   \begin{equation}
    \label{EqnSecurityPrice}
S^k_t=\condexp{\int_t^T A^k_s \psi_s \,  \nu(ds)}\,.
  \end{equation}
  \end{theorem}

The proof of this theorem runs as follows. In a first step, we establish the existence of an Arrow--Debreu equilibrium. In the current time--additive setup, this is a result by \cite{Dana02}. We extend her result by showing that in our smooth economy the equilibrium consumption price $\psi$  and the allocation $\left(c^i\right)_{i=1,\ldots,I}$ are analytic functions of time and the state variable.
It is well known that one can implement the Arrow--Debreu equilibrium as a Radner equilibrium if one has dynamically complete markets. With \emph{nominal} assets, this is more or less trivial (see  \cite{DuffieHuang85} and \cite{Huang87}). Here, our assets pay real dividends, and the completeness depends on the endogenous consumption price $\psi$ and cannot be assumed exogenously.

The natural candidates for our asset prices are, of course, the present values of their future dividends as in \rref{EqnSecurityPrice}. We have to show dynamic completeness then. We do this by proving that the (local) linear independence of the dividends at maturity $T$ carries over to the volatility matrix of asset prices. This yields dynamic  completeness. This step needs the intermediate mathematical result that our candidate security prices are analytic functions of time and state variable.

The implementation of the Arrow--Debreu equilibrium as a Radner equilibrium is then standard.

\subsection{Existence of an Analytic Arrow--Debreu Equilibrium}

We quickly recall the notions of classical General Equilibrium Theory.
An allocation is an element $\left(c^i\right)_{i=1,\ldots,I} \in \cX_+^I$. Is is feasible if we have $\sum_{i=1}^I c^i \le \varepsilon$. A price  is a nonnegative, optional process $\psi\in\cX_+$.  It defines a continuous linear price functional $\Psi(c)=\E \int_0^T  c_t \psi_t \, \nu(dt)$ on $\cX$.

An Arrow--Debreu equilibrium consists of a feasible allocation
$\left(c^i\right)_{i=1,\ldots,I}$ and
a price $\psi$ such that  $c^i$ is budget--feasible and optimal for all agents $i=1, \ldots, I$, i.e. $\Psi(c^i)\le\Psi(\varepsilon^i)$,  and for all consumption plans $c \in \cX_+$ the relation
$U^i(c)>U^i(c^i)$ implies $\Psi(c)>\Psi(\varepsilon^i)$.

Existence and uniqueness of Arrow--Debreu equilibria in our separable setting have been clarified by \cite{Dana93}. We recall her existence result and  show the additional refinement that equilibrium price and consumption plans are analytic functions of time and the state variable on $(0,T)\times \mathbb R^K$.

\begin{theorem}\label{ThmArrowDebreuEquilibrium}
  Under Assumptions \ref{AssUtilityGeneral}, \ref{AssEntitlementAnalytic}, and \ref{margfel} i), there exists an Arrow--Debreu equilibrium $\left(\psi, \left(c^i\right)_{i=1,\ldots,I}\right)$ such that
  \begin{align*}
    \psi_t &= \psi(t,X_t)\\
    c^i_t &= c^i(t,X_t)
  \end{align*}
  for some continuous functions $$ \psi, c^i : [0,T] \times \mathbb R^K \to \mathbb R_+$$ that are analytic on $(0,T)\times \mathbb R^K$.
\end{theorem}

\begin{proof}
  By \cite{Dana93}, there exists an equilibrium $(\psi, (c^i))$ with $\psi >0$ $P\otimes\nu$--a.s.  and the allocation $(c^i)$  is the solution of  the social planner problem
$$\max_{c \in \cX_+^I, \sum c^i \le \varepsilon} \sum \lambda^i U^i(c^i)$$
for some $\lambda^i >0$\footnote{$\lambda^i =0$ is not possible. This is already implicit in Dana's proof. Here is another argument based on our Assumption \ref{margfel}. For, if, say, $\lambda^1=0$, then $c^1=0$ (by Negishi). By the strict monotonicity of utility functions, $c^1=0$ is an equilibrium demand only if wealth is zero, i.e. $E\int_0^T \psi_t \varepsilon^1_t \nu(dt)=0$. But by Assumption \ref{margfel} and the Inada assumption, $\varepsilon^1>0 $ $P\otimes\nu$--a.s. Hence $E\int_0^T \psi_t \varepsilon^1_t \nu(dt)>0$, a contradiction.}.

 As we have separable utility functions, the social planner's problem can be solved point-- and state--wise; we thus look at the real--valued problem
$$ v(t,x):=\max_{\substack{ \sum_{i=1}^I x^i=x\\ x^i \ge 0, i=1,\ldots, I}}
 \sum_{i=1}^I \lambda^i u^i(t,x^i) \,.$$
By Assumption \ref{AssUtilityGeneral}, the unique solution of the above real--valued maximization problem is characterized by the equations
\begin{align}
  \lambda^i \frac{\partial u^i}{\partial c} \left(t,x^i\right) &= \mu\label{EqnSocialPlanner1}\\
  \sum_{i=1}^I x^i &= x\label{EqnSocialPlanner2}
\end{align}
for some Lagrange parameter $\mu>0$. By \cite{Dana93}, Proposition 2.1, the solution of the above equations is given by continuous functions
$x^i, \mu : [0,T] \times \mathbb R_+ \to \mathbb R_+ $ of $(t,x)$. By the Analytic Implicit Function Theorem and Assumption \ref{AssUtilityGeneral}, these are even analytic  on  $(0,T) \times (0,\infty)$ (see also \cite{AndersonRaimondo08}, page 881). By \cite{Dana93}, we have $c^i_t=x^i\left(t,\varepsilon_t\right)$ and $\psi_t = \mu\left(t,\varepsilon_t\right)$. As aggregate endowment is a function of time and state variable that is continuous on $[0,T] \times \mathbb R_+$ and analytic on $(0,T)\times \mathbb R_+$ (Assumptions \ref{AssEntitlementAnalytic} and \ref{AssDividendAnalytic}), the result follows.
\end{proof}

\subsection{Analytic Security Prices}

We can now conclude the proof of our main theorem \ref{ThmMain}. In particular, we assume Assumptions 1 through 5, 6 ii) and 7.

The natural candidates for security prices are, of course, the present values of their future dividends, or \rref{EqnSecurityPrice}. The corresponding gain processes are then
\begin{equation}\label{EqnGain}
  G^k_t = S^k_t + \int_{[0,t)} A^k_s  \psi_s \nu(ds) ,\qquad(0\le t \le T)\,.
\end{equation}
It is of essential importance for our development that these expectations are themselves analytic functions of time and state variable jointly.

\begin{theorem}\label{ThmAnalyticAssets}
  Define $S$ by \rref{EqnSecurityPrice}. \frederik{Under Assumptions 1 to 5, \ref{margfel} ii), and 7,} there exist continuous functions
  $s: [0,T] \times \mathbb R^K \to \mathbb R_+$ that are analytic on $(0,T)\times\mathbb R^K$ and
  $$S_t = s(t,X_t)\,.$$
  The first derivatives with respect to $x$, $\frac{\partial s}{\partial x_l}$ are continuous on $[0,T] \times \mathbb R^K$ and we have
  $$\lim_{t \uparrow T} \frac{\partial s}{\partial x_l}(t,x)=\frac{\partial s}{\partial x_l} (T,x)=\frac{\partial g}{\partial x_l}(T,x)$$
\end{theorem}

\begin{proof}
As $X$ is a Markov process, $ \frederik{S}^k_t$ is a function $\frederik{s}^k$ of time $t$ and state $X_t$, that is
$$ \frederik{s}^k(t,x)= \E\left[ \int_{\frederik{t}}^T m^k(s,X_s)   \,  \nu(ds) | X_t=x \right]
$$
for $m^k(t,x)=g^k(t,x)\psi(t,x)$\frederik{ (the existence of such a function $\psi$ was established in Theorem~\ref{ThmArrowDebreuEquilibrium})}.

By Assumption \ref{margfel} ii) $m^k$ are bounded functions, and hence $s^k$ is bounded as well.
 By Theorem 5.3 of Chapter 6 on p. 148 of \cite{FriedmanSDE1}, or alternatively, by \cite{HeathSchweizer00} $s^k$ is a classical ${\cal C}^{1,2}$--solution of  the Cauchy problem
 $$ -  \frac{\partial}{\partial t } u +  \cL u = m^k$$ with boundary condition $s^k(T,x)=m^k(T,x)$.

Now let $d(x)=\exp(-l(x))$ for a smooth function $l$ on $\mathbb R^n$ that satisfies $l(x)=\|x\|$ for $\|x\|\ge 1$ and set $u(t,x)=s^k(t,x) d(x)$. Then $u(T,x)=m^k(T,x) d(x) \in L^2(\mathbb R^n)$ and $u$ solves the modified Cauchy problem
 $$ -  \frac{\partial}{\partial t } u +  \tilde{\cL } u = f$$
for some suitable function $f$ and elliptic operator $\tilde{\cL}$ with
$$\tilde{\cL} (v d) = d \cL v \qquad (v \in {\cal C}^\infty (\mathbb R^n) \frederik{)}\,.$$

As $\tilde{\cL}$ is a sectorial operator on $L^2(\mathbb R^n)$\frederik{ (see, for instance, Proposition 3.1.17 of \cite{Lunardi95} and Theorem 5.2 in Section 2.5 of \cite{Pazy83})},
we can apply standard results on evolution equations in Banach spaces to conclude that $u$, and then $s^k$, is analytic in time $t$. (For example, the Corollary on page 209 in \cite{FriedmanPDE} applies. We can take the Sobolev space $X={\mathbb W}^2_2$ as the Banach space there. Our operator $\cal L$ has $X$ as its domain. Hence, condition (E1) there is satisfied. Condition (E2) requires that the resolvent of the Markov process exists in some complex sector around zero. However, this has been proven in Eq. (2.11), Theorem 1 of \cite{Yosida59}, or see the references above. (E3) in that book is automatically satisfied as our operator is independent of time.)

Finally, note that $s^k$ is analytic in $(x,y)$ by Theorem~1.2 in Part~3, Chapter~1 of \cite{FriedmanPDE} and recall that functions which are bounded and separately analytic are jointly analytic (a result of \cite{Osgood1899}).

The continuous differentiability of $s^k$ with respect to the second argument $x$ follows from Theorem 10.3 on p.~143 of \cite{FriedmanPDE}.
\end{proof}

\frederik{

\begin{remark}
Perhaps the reader is wondering whether one might obtain analyticity directly by invoking an appropriate general result from the theory of partial differential equations. To be sure, the function $s$ in Theorem~\ref{ThmAnalyticAssets} solves an inhomogeneous parabolic differential equation (see again Theorem 5.3 in Chapter 6 of \cite{FriedmanSDE1}, and indeed there does exist a body of literature on analyticity of solutions to linear second-order partial differential equations. For example, \cite{DeGiorgiCattabriga71} proved that for the special case of (space-time) dimension $2$ (i.e. in our setting $K=1$) and constant coefficients, any solution to an inhomogeneous linear differential equation with analytic right-hand side (in our case, this is the product of the pricing density, as a function of $X$ and the dividend function $g$) is again analytic, hence if our process $X$ is just a one-dimensional Brownian motion with drift, then analyticity of $s$ follows. However, as already conjectured by \cite{DeGiorgiCattabriga71}, this result fails to hold in general for higher dimensions. A counterexample is the heat equation with two-dimensional Laplacian, as was proved by \cite{Piccinini73}, using a right-hand side which grows, while still analytic, at enormous pace. Therefore, one cannot, in general, dispense of growth conditions like those imposed by \cite{AndersonRaimondo08} in their Theorem~B.4. These negative results have been generalized in numerous papers, in particular by \cite{Hoermander73} (for the case of partial differential operators with constant coefficients) and \cite{OleinikRadkevic73} (for the case of partial differential operators with analytic coefficients).

\cite{OleinikRadkevic82} do give sufficient conditions for analyticity of solutions of all solutions --- in the distribution sense --- of inhomogeneous second-order linear partial differential equations with analytic right-hand side. However, these conditions are limited to the two-dimensional case (which would mean $K=1$ in our setting) and involve the assumption that the equation can be transformed into another partial differential equations where there are second-order diagonal terms in both variables (in our case this would imply a second-order time-derivative), see Theorem~2 and Theorem~3 of \cite{OleinikRadkevic82}.
\end{remark}

}

\subsection{Dynamically Complete Markets}

\begin{theorem}
  \frederik{Under Assumptions 1 to 5, \ref{margfel} ii), and 7,} the market $(S,A,\psi)$ is dynamically complete.
\end{theorem}

\begin{proof}
  By Assumption \ref{AssIndependentPayoffs}, $\psi>0$ and the fact that $X$ has full support, we have $S^0_t>0$ a.s. Hence, we can take asset $0$ as a num\'{e}raire.
  Define $$R^k_t=\frac{S^k_t}{S^0_t}\,.$$
  By Theorem \ref{ThmAnalyticAssets}, $R^k_t=r^k(t,X_t)$ for continuous functions
  $r^k: [0,T] \times \mathbb R^K \to \mathbb R_+$ that are analytic on $(0,T)\times\mathbb R^K$, $k=1,\ldots,K$.

  After this change of num\'{e}raire, we have a riskless asset (with interest rate $0$, of course) and $K$ risky assets, as many as independent Brownian motions.

  The asset market is dynamically complete if the volatility matrix is a.s. invertible (see, e.g., \cite{KaratzasShreve98}, Theorem 1.6.6)\footnote{To apply this result, we check quickly that  the asset market is also standard in the sense of \cite{KaratzasShreve98}:
  by construction (\rref{EqnSecurityPrice}), the gain processes are martingales; hence, our market is arbitrage--free. As our state--price deflator $\psi$ is in $\Psi$, also the martingale condition in \cite{KaratzasKou98} is satisfied.}. By It\^{o}'s lemma, the volatility matrix is given by
  $  I(t,x) Dr(t,x) \sigma(t,x) $  where $Dr$ is the Jacobian matrix of $r$ and
   $I$ the triangular matrix
   $$
  I(t,x)=
  \begin{pmatrix}
    \frac{1}{r_1(t,x)} & \ldots    & 0 \\
    \vdots & \ddots & \vdots \\
    0 & \ldots & \frac{1}{r_K(t,x)} \\
  \end{pmatrix}\,.$$

  Now suppose that the volatility matrix has determinant $0$ on a set of positive Lebesgue measure. By analyticity and Theorem B.3 in \cite{AndersonRaimondo08}, we conclude that the determinant vanishes everywhere on $(0,T)\times \mathbb R^K$. As $Dr$, $r$,  and $\sigma$ are continuous on $[0,T]$, it then follows that
  $$ \det I(T,x) Dr(T,x)\sigma(T,x)  = 0\,.$$
  (For $Dr$ and $r$, this is Theorem \ref{ThmAnalyticAssets}.)
  As $\sigma$ has full rank by Assumption \ref{AssSDE} and $I(T,x)$ is triangular, we conclude that
  $$ \det Dr(T,x)  = 0 \,.$$ But $r(T,x)=g(T,x)/g^0(T,x)=h(x)$, so
  $$ \det Dr(T,x)  \not= 0$$ on a set of positive measure by Assumption \ref{AssIndependentPayoffs}.
  This contradiction shows that the volatility matrix is invertible a.s. We conclude that the market $(S,A,\psi)$ is dynamically complete.

\end{proof}

With dynamically complete asset markets, it is a standard argument to show that the Arrow--Debreu equilibrium can be implemented as a Radner equilibrium. The basic argument is as in \cite{DuffieHuang85}, translated to our more complex setting, see also \cite{DanaJeanblanc03}, Theorem 7.1.10 (apply this theorem to the asset market with asset $0$ as num\'{e}raire).

\bibliography{bibliography_riedelherzberg}

\end{document}